\documentclass[11pt]{article}

\usepackage{latexsym,mathptm}
\usepackage{epsfig}
\usepackage{url}
\usepackage{rotating}
\usepackage{graphicx}
\usepackage{float}
\usepackage{booktabs}

\usepackage{amsmath}
\usepackage{amssymb}
\usepackage{fullpage}
\usepackage{amsthm}
\usepackage{epsf}
\usepackage{graphicx}
\usepackage{color}
\usepackage{enumerate}

\RequirePackage[colorlinks,linkcolor=blue,filecolor = blue,citecolor = blue, urlcolor =blue]{hyperref}
\RequirePackage{hyperref}

\newcommand{\ignore}[1]{}

\renewcommand{\l}{\ell}

\newtheorem{theorem}{Theorem}

\newtheorem{corollary}[theorem]{Corollary}
\newtheorem{definition}{Definition}

\newtheorem{remark}{Remark}

\renewcommand{\le}{\leqslant}
\renewcommand{\ge}{\geqslant}

\newcommand{\eps}{\varepsilon}
\renewcommand{\epsilon}{\varepsilon}

\newcommand{\vnote}[1]{}
\newcommand{\anote}[1]{}

\newcommand{\F}{\mathbb{F}}
\newcommand{\R}{\mathbb{R}}

\newcommand{\A}{{\mathcal{A}}}

\title{Storage Enforcement with Kolmogorov Complexity and List Decoding}

\author{Mohammad Iftekhar Husain\and Steve Ko\and Atri Rudra\footnote{Supported in part by NSF CAREER grant CCF-0844796}\and Steve Uurtamo\footnotemark[1]}

\date{Department of Computer Science and Engineering,\\
University at Buffalo, SUNY,\\
Buffalo, NY, 14214.\\
{\tt \{imhusain,stevko,atri,uurtamo\}@buffalo.edu}}

\begin{document}
\maketitle

\setcounter{page}{0}
\thispagestyle{empty}

\begin{abstract}
We consider the following problem that arises in outsourced storage: a user stores her data $x$ on a remote server but wants to audit the server at some later point to make sure it actually did store $x$. The goal is to design a (randomized) verification protocol that has the property that if the server passes the verification with some reasonably high probability then the user can rest assured that the server is storing $x$.
While doing this, we need to minimize the user's storage and the total amount of communication between the server and the user. This is the data possession problem and is closely related to the problem of obtaining a 'proof of retrievability'. Existing schemes with provable guarantees mostly use cryptographic primitives and can only guarantee that the server is storing a constant fraction of the amount of information that it is supposed to store.

In this work we present an optimal solution (in terms of the user's storage and communication) while at the same time ensuring that a server that passes the verification protocol with any reasonable probability will store, to within a small \textit{additive} factor, $C(x)$ bits of information, where $C(x)$ is the plain Kolmogorov complexity of $x$. (Since we cannot prevent the server from compressing $x$, $C(x)$ is a natural upper bound.) The proof of security of our protocol combines Kolmogorov complexity with list decoding and unlike previous work that relies upon cryptographic assumptions, we allow the server to have unlimited computational power. To the best of our knowledge, this is the first work that combines Kolmogorov complexity and list decoding.

Our framework is general enough to capture extensions where the user splits up $x$ and stores the fragment across multiple servers and our verification protocol can handle non-responsive servers and colluding servers. As a by-product, we also get a proof of retrievability. Finally, our results also have an application in `storage enforcement' schemes, which in turn have an application in trying to update a remote server that is potentially infected with a virus.

\end{abstract}


\newpage

\section{Introduction}

We consider the following problem: A string $x\in [q]^k$ is held by the client/user, then transmitted to a remote server (or broken up into several pieces and transmitted to several remote servers). At some later point, after having sent $x$, the client would like to know that the remote server(s) can reproduce $x$. At the very least the client would like to know that as much space was committed by the remote servers to storage as was used by $x$. The former problem has been studied under the name of \textit{proof of retrievability} (cf.~\cite{JuelsPOR,BowersPOR,DVW09}) 
and/or \textit{data possession} (cf.~\cite{AteniesePDP,AtenieseSEPDP,CurtmolaMRPDP,ErwayDPDP}). The latter has been studied under the name of \textit{storage enforcing commitment}~\cite{GJM02}. (In these problems there is only one server.) Given the greater prevalence of outsourcing of storage (e.g. in ``cloud computing"), such a verification procedure is crucial for auditing purposes to make sure that the remote servers are following the terms of their agreement with the client.

The naive solution would be for the client to store $x$ locally and during the verification stage, ask the server(s) to send back $x$. However, in typical applications where storage is outsourced, the main idea behind the client shipping $x$ off to the servers is that it does not want to store $x$ locally. Also, asking the server(s) to send back the entire string $x$ is not desirable because of the communication cost. Thus, we want to design a verification protocol such that
the bandwidth $c$ used for the challenge and challenge response is low compared with the size of $x$ (which we will denote by $|x|$) and such that the local storage $m$ for the challenger should be low (far less than $|x|$). It is not hard to see that a deterministic protocol will not be successful. (The server could just store what the client stores and compute the correct answer and send it back.) Thus, we need a randomized verification protocol which catches ``cheating" servers with high probability. In other words, if the server(s) passes the verification with probability at least $\eps$, then the server(s) must necessarily have to have stored a large portion of $x$.

There are many other parameters for such a protocol and we highlight some of them here. First, one has to quantify the ``largeness" of the portion of $x$ that the server(s) are forced to store. Typically, existing results show that if server(s) pass the verification process with probability $\eps$, then it stores (or is able to re-create) a \textit{constant} fraction of $x$. Second, one has to fix the assumptions on the computational power of the client and the servers. Typically, all of the honest parties are constrained to be polynomial time algorithms. In addition, one might want to minimize the query complexity of an honest server while replying back to a challenge (e.g.~\cite{DVW09}). Alternatively, one might only allow for the client and server to use one pass low space data stream algorithms~\cite{CTY10}.\footnote{\cite{CTY10} actually looks at a more general problem where the client wants to verify whether the server has correctly done some computation on $x$. In our case, one can think of the outsourced computation as the identity function.} Third, one needs to decide on the computational power allowed to a cheating server. Typically, the cheating server is assumed to not be able to break cryptographic primitives. Finally, one needs to decide whether the client can make an unbounded number of audits based on the same locally stored data. Many of the previous work based on cryptographic assumptions allows for this. 

Before we move on to our result, we would like to make a point that does not seem to have been made explicitly before. Note that we cannot prevent a server from compressing their copy of the string, or otherwise representing it in any reversible way. (Indeed, the user would not care as long as the server is able to recreate $x$.) This means that a natural upper bound on the amount of storage we can force a server into is $C(x)$, the (plain) Kolmogorov complexity of $x$, which is the size of the smallest (algorithmic) description of $x$.

In all of our results, if the server(s) pass the verification protocol with probability $\eps>0$, then they provably have to store $C(x)$ bits of data up to a very small \textit{additive} factor. Cheating servers are allowed to use arbitrarily powerful algorithms (as long as they terminate) while responding to challenges from the user.\footnote{The server(s) can use different algorithms for different strings $x$ but the algorithm cannot change across different challenges for the same $x$.} Further, every honest server only needs to store $x$ (or its portion of $x$). In other words, unlike some existing results, our protocol does not require the server(s) to store a massaged version of $x$. In practice, this is important as we do not want our auditing protocol to interfere with ``normal" read operations. However, unlike many of the existing results based upon cryptographic primitives, our protocol can only allow a number of audits that is proportional to the user's local storage.\footnote{An advantage of proving security under cryptographic assumptions is that one can leverage existing theorems to prove additional security properties of the verification protocol. We do not claim any additional security guarantees other than the ability of being able to force servers to store close to $C(x)$ amounts of data.} For a more detailed comparison with existing work, especially in the security community, see section \ref{Appendix:tables}.

Our main result, which might need the user and honest servers to be exponential time algorithms, provably achieves the optimal local storage for the user and the optimal total communication between the user and the server(s). With slight worsening of the storage and communication parameters, the user and the honest servers can work with single pass logarithmic space data stream algorithms. Finally, for the multiple server case, we can allow for arbitrary collusion among the servers if we are only interested in checking if at least one sever is cheating (but we cannot identify it). If we want to identify at least one cheating server (and allow for servers to not respond to challenges), then we can handle a moderately limited collusion among servers.




Even though we quantitatively improve existing results (at least on most parameters), we believe that the strongest point of the paper is the set of techniques that we use. In particular, our proofs naturally combine the notions of Kolmogorov complexity and list decoding. These two notions have been used widely in computational complexity but to the best of our knowledge ours is the first one to use both at the same time. Next we present more details on our techniques.


\paragraph{Our techniques:} In this part of the paper, we will mostly concentrate on the single server case.

To motivate our techniques, let us look at the somewhat related problem of designing a communication protocol for the set/vector equality problem. In this problem Alice is given a string $x$ and Bob is given $y$ and they want to check if $x=y$ with a minimum amount of communication. (One could think of $y$ as the ``version" of $x$ that Bob (aka the server) stores for $x$.)\footnote{The main difference is that in this communication complexity problem, Alice and Bob are both trying to help each other while in our case Bob could be trying to defeat Alice.} The well-known fingerprinting protocol picks a random hash function $h$ (using some public randomness) and Alice sends Bob $h(x)$ who checks if $h(x)=h(y)$. Two classical hash functions $h(x)$ correspond to taking $x$ mod a random prime (a.k.a. the Karp Rabin fingerprint~\cite{KarpRabin}) and thinking of $x$ as defining a polynomial $P_x(Y)$ and evaluating $P(Y)$ at a random field element.

It is also well-known that both of the hashes above are a special case of the following class of hash functions. Let $H:[q]^k\rightarrow [q]^n$ be an error-correcting code with large distance. Then a random hash involves picking a random $\beta\in [n]$ and defining $h_{\beta}(x)=H(x)_{\beta}$. (The Karp Rabin hash corresponds to $H$ being the so called ``Chinese Remainder Theorem" (CRT) code and the polynomial hash corresponds to $H$ being the Reed-Solomon code.) Our result needs $H$ to have good \textit{list decodable} properties. (We do not need an algorithmic guarantee, as just a combinatorial guarantee suffices.)

The user picks a random $\beta$ and stores $(\beta,H(x)_{\beta})$. During the verification phase, it sends $\beta$ to the server and asks it to compute $H(x)_{\beta}$. If the server's answer $a\neq H(x)_{\beta}$ it rejects, otherwise it accepts.
We now quickly sketch why the server cannot get away with storing a vector $y$ such that $|y|$ is smaller than $C(x)$ by some appropriately small additive factor. Since we are assuming that the server uses an algorithm $\A$ to compute its answer $\A(\beta,y)$ to the challenge $\beta$, if the server's answer is accepted with probability at least $\eps$, then note that the vector $(\A(\beta,y))_{\beta\in [n]}$ differs from $H(x)$ in at most $1-\eps$ fraction of positions. Thus, if $H$ has good list decodability, then (using $\A$) one can compute a list $\{x_1,\dots,x_L\}$ that contains $x$. Finally, one can use $\log{L}$ bits of advice (in addition to $y$) to output $x$. This procedure then becomes a description for $x$ and if $|y|$ is sufficiently smaller than $C(x)$, then our description will have size $<C(x)$, which contradicts the definition of $C(x)$.

The trivial solution for the multiple server case would be to run independent copies of the single server protocol for each server.
Our techniques very easily generalize to the multiple server case, which lead to a protocol whose user storage requirement matches that of the single server case, which is better than the trivial solution. For this generalization, we need $H$ to be a linear code (in addition to having good list decodability).  Further, by applying a systematic code on the hashes for the trivial solution and only storing the parity symbols, one can also handle the case when some servers do not respond to challenges while using user storage that is somewhere in between that of the trivial solution and the single server protocol.

From a more practical point of view, Reed-Solomon and CRT codes have good list decodability, which implies that our protocols can be implemented using the classical Karp Rabin and polynomial hashes.

\paragraph{Another application:} We believe that our results and especially our techniques should be more widely applicable. Next we briefly mention an application of our result to a practical problem that was pointed out to us by Dick Lipton.

Assume that the user wishes to update an operating system on a remote computer, but is concerned that a virus on the remote machine may be listening at the network device, intercepting requests and answering them without actually installing the operating system, or by installing the operating system and cleverly reinserting the virus while doing so.

Our single server protocol can be used, along with a randomly-generated string $x$, (which will, with constant probability greater than $1/2$, have $C(x)\ge (|x|-O(1))$)\footnote{The exact bound is very strong. $C(x)\ge (|x|-r)$ for at least a $(1-1/{2^r})$-fraction of the strings of length $|x|$.}, to force the remote machine to first store a string of length close to $|x|$ (overwriting any virus as long as $|x|$ is chosen large enough), second, to correctly answer a verification request (a failure to answer the request proves that the install has failed), and third (now that there is no room for a virus) to install the new operating system.

\section{Preliminaries}


We use $\mathbb{F}_{q}$ to denote the
finite field over $q$ elements. We also use $[n]$ to denote the set
$\{1,2,..,n\}$. Given any string $x\in\F_q^*$, we use $|x|$ to denote the length of $x$ in bits. Additionally, all logarithms will be base 2 unless otherwise specified.

\subsection{Verification Protocol}

We now formally define the different parameters of a verification protocol. We use $U$ to denote the user/client. We assume that $U$ wants to store its data $x\in\F_q^k$ among $s$ service providers $P_1,\dots,P_s$. In the pre-processing step of our setup, $U$ sends $x$ to $\mathcal P=\{P_1,\dots,P_s\}$ by dividing it up equally among the $s$ servers -- we will denote the chunk sent to server $i\in [s]$ as $x_i\in\F_q^{n/s}$.\footnote{We will assume that $s$ divides $n$. Further, in our results for the case when $H$ is a linear code, we do not need the $x_i$'s to have the same size, only that $x$ can be paritioned into $x_1,\dots,x_s$. We will ignore this possibility for the rest of the paper.} Each server is then allowed to apply any computable function to its chunk and to store a string $y_i\in\F_q^*$. Ideally, we would like $y_i=x_i$. However, since the servers can compress $x$, we would at the very least like to force $|y_i|$ to be as close to $C(x_i)$ as possible. For notational convenience, for any subset $T\subseteq [s]$, we denote $y_T$ ($x_T$ resp.) to be the concatenation of the strings $\{ y_i\}_i\in T$ ($\{x_i\}_{i\in T}$ resp.).

To enforce the conditions above, we design a protocol. We will be primarily concerned with the amount of storage at the client side and the amount of communication and want to minimize both simultaneously while giving good verification properties. The following definition captures these notions. (We also allow for the servers to collude among each other.)

   \begin{definition}
Let $s,c,m\ge 1$ and $0\le r\le s$ be integers, $0\le \rho\le 1$ be a real and $f:[q]^*\rightarrow \R_{\ge 0}$ be a function. Then an $(s,r)$-{\em party verification protocol} with resource bound $(c,m)$ and verification guarantee $(\rho,f)$ is a randomized protocol with the following guarantee.
For any string $x\in [q]^k$, $U$ stores at most $m$ bits and communicates at most $c$ bits with the $s$ servers. At the end, the protocol either outputs a $1$ or a $0$. Finally, the following is true for any $T\subseteq [s]$ with $|T|\le r$: If the protocol outputs a $1$ with probability at least $\rho$, then assuming that every server $i\in [s]\setminus T$ followed the protocol and that every server in $T$ possibly colluded with one another, we have $|y_T|\ge f(x_T)$.
\end{definition}

We will denote a $(1,1)$-party verification protocol as a one-party verification protocol. (Note that in this case, the single server is allowed to behave arbitrarily.)

All of our protocols will have the following structure: we first pick a family of ``keyed" hash functions. The protocol will pick  random key(s) and store the corresponding hash values for $x$ (along with the keys) during the pre-processing step. During the verification step, $U$ sends the key(s) as challenges to the $s$ servers. Throughout this paper, we will assume that each server $i$ has a computable algorithm $\A_{x,i}$ such that on challenge $\beta$ it returns an answer $\A_{x,i}(\beta,y_i)$ to $U$. The protocol then outputs $1$ or $0$ by applying a (simple) boolean function on the answers and the stored hash values.

   \subsection{List Decodability}

We begin with some basic coding definitions. 
An (error-correcting) \textit{code} $H$ with \textit{dimension} $k\ge 1$ and block length $n\ge k$ over an \textit{alphabet} of size $q$ is any function $H: [q]^k \rightarrow [q]^n$.
   A \textit{linear code} $H$ is any error-correcting code that is a linear function, in which case we correspond $[q]$ with $\F_q$.
   A \textit{ message} of a code $H$ is any element in the domain of $H$.
   A \textit{ codeword} in a code $H$ is any element in the range of $H$. 
   
The \textit{ Hamming distance} $\Delta(x,y)$ of two same-length strings is the number of symbols in which they differ. The \textit{ relative distance} $\delta$ of a code is $\min_{x\neq y}\frac{\Delta(x,y)}{n}$, where $x$ and $y$ are any two different codewords in the code.

   \begin{definition} A $(\rho,L)$ {\em list-decodable code} is any error-correcting code such that for every codeword $e$ in the code, the set $E'$ of codewords that are Hamming distance $\rho n$ or less from $e$ is always $L$ or fewer.
\end{definition}

   Geometric intuition for a $(\rho,L)$ list-decodable code is that it is one where Hamming balls of radius $\rho n$ centered at arbitrary vectors in $[q]^n$ always contain $L$ or fewer other codewords.

   For the purposes of this paper, we only consider codes $H$ that are members of a family of codes $\mathcal{H}$, any one of which can be indexed by $(k,n,q,\rho n)$. Note that not all values of the tuple $(k,n,q,\rho n)$ need represent a code in $\mathcal{H}$, just that each $H_i\in \mathcal{H}$ has a distinct such representation.

   \subsection{Plain Kolmogorov Complexity}

   \begin{definition}
   The {\em plain Kolmogorov Complexity} $C(x)$ of a string $x$ is the minimum sum of sizes of a compressed representation of $x$, along with its decoding algorithm $D$, and a reference universal Turing machine $T$ that runs the decoding algorithm.
   \end{definition}

   Because the reference universal Turing machine size is constant, it is useful to think of $C(x)$ as simply measuring the amount of inherent (i.e. incompressible) information in a string $x$.

   Most strings cannot be compressed beyond a constant number of bits. This is seen by a simple counting argument. $C(x)$ measures the extent to which this is the case for a given string. 

\section{One Remote Party Result}

We begin by presenting our main result for the case of one server ($s=1$) to illustrate the combination of list decoding and Kolmogorov complexity. In the subsequent section, we will generalize our result to the multiple server case.

\begin{theorem}
\label{thm:main-single}
For every computable error-correcting code $H: [q]^k \rightarrow [q]^n$ that is $(\rho,L)$ list-decodable, there exists a one-party verification protocol with resource bound $(\log{n} + \log{q},\log{n} + \log{q})$ and verification guarantee $(1-\rho,f)$, where  for any $x\in [q]^k$, $f(x)=C(x) - \log (qLn^3) - 2\log\log (qn) -c_0$, for some fixed constant $c_0$.\footnote{The contribution from the encoding of the constants in this theorem to $c_0$ is 2. For most codes, we can take $c_0$ to be less than a few thousand. What is important is that the contribution from the encoding is independent of the rest of the constants in the theorem. Although we do not explicitly say so in the body of the theorem statements, this important fact is true for the rest of the results in this paper as well.}
\end{theorem}

\begin{proof}
We begin by specifying the protocol. In the preprocessing step, the client $U$ does the following on input $x\in [q]^k$:
\begin{enumerate}
\item Generate a random $\beta\in [n]$.
\item Store $(\beta,\gamma=H(x)_{\beta})$ and send $x$ to the server.
\end{enumerate}

The server, upon receiving $x$, saves a string $y\in [q]^*$. The server is allowed to use any computable function to obtain $y$ from $x$.

During the verification phase, $U$ does the following:
\begin{enumerate}
\item It sends $\beta$ to the server.
\item It receives $a\in [q]$ from the server. ($a$ is supposed to be $H(x)_{\beta}$.)
\item It outputs $1$ (i.e. server did not ``cheat") if $a=\gamma$, else it outputs a $0$.
\end{enumerate}

We assume that the server, upon receiving the challenge, uses a computable function $\A_x:[n]\times [q]^*\rightarrow [q]$ to compute $a=\A_x(\beta,y)$ and sends $a$ back to $U$.

The claim on the resource usage follows immediately from the protocol specification. Next we prove its verification guarantee by contradiction. Assume that $|y|<f(x)\stackrel{def}{=} C(x) - \log (qLn^3) - 2\log\log (qn) -c_0$ and yet the protocol outputs $1$ with probability at least $1-\rho$ (over the choice of $\beta$). Define $z=(\A_x(\beta,y))_{\beta\in [n]}$. Note that by the claim on the probability,  $\Delta(z,H(x))\le \rho n$. We will use this and the list decodability of the code $H$ to prove that there is an algorithm with description size $< C(x)$ to describe $x$, which is a contradiction. To see this, consider the following algorithm that uses $y$ and an advice string $v\in\{0,1\}^{|{L}|}$:
\begin{enumerate}
\item Compute a description of $H$ from $n,k,\rho n$ and $q$.
\item Compute $z=(\A_x(\beta,y))_{\beta\in [n]}$.
\item By cycling through all $x\in [q]^k$, retain the set $\mathcal L\subseteq [q]^k$ such that for every $u\in \mathcal L$, $\Delta(H(u),z)\le \rho n$.
\item Output the $v$th string from $\mathcal L$.
\end{enumerate}

Note that since $H$ is $(\rho,L)$-list decodable, there exists an advice string $v$ such that the algorithm above outputs $x$. Further, since $H$ is computable, there is an algorithm $\mathcal E$ that can compute a description of $H$ from $n,k,\rho n$ and $q$. (Note that using this description, we can generate any codeword $H(u)$ in step 3.) Thus, we have description of $x$ of size $|y|+|v|+|\A_x|+|\mathcal E|+ (3\log{n}+\log{q}+2\log\log n+2\log\log q+2)$ (where the last term is for encoding the different parameters\footnote{We use a simple self-delimiting encoding of $q$ and $n$, followed immediately by $k$ and $\rho n$ in binary, with the remaining bits used for $v$. A simple self-delimiting encoding for a positive integer $u$ is the concatenation of: $( \lceil\log(|u|)\rceil$ in unary, $0$, $|u|$ in binary, $u$ in binary). We omit the description of this encoding in later proofs.}),
which means if $|y| < C(x)-|v|-|\A_x|-|\mathcal E|-(3\log{n}+\log{q}+2\log\log n+2\log\log q+2) = f(x)$, then we have a description of $x$ of size $<C(x)$, which is a contradiction.

\end{proof}

The one unsatisfactory aspect of the result above is that if $H$ is not polynomial time computable, then Step 2 in the pre-processing step for $U$ is not efficient. Similarly, if the sever is not cheating (and e.g. stores $y=x$), then it cannot also compute the correct answer efficiently. We will come back to these issues when we instantiate $H$ by an explicit code such as Reed-Solomon.

\begin{remark}
If $H$ has relative distance $\delta$, then note that if our protocol has verification guarantee $(1-\delta/2-\epsilon,C(x)-\log(qn^3)-2\log\log(qn)-c_0)$ for some fixed constant $c_0$, then $y$ has enough information for the server to compute $x$ back from it. (It can use the same algorithm to compute $x$ from $y$ detailed above, except it does not need the advice string $v$, as in Step 3, we will have $\mathcal L=\{x\}$.) For the more general case when $H$ is $(\rho,L)$-list decodable and our protocol has verification guarantee $(1-\rho,C(x)-\log{qLn^3}-2\log\log{qn}-c_0)$, then $y$ has enough information for the server to compute a list $\mathcal{L}\supseteq \{x\}$ with $|\mathcal L|\le L$. The client, if given access to $\mathcal L$, can use its local hash to pick $x$ out of $\mathcal L$ with probability at least $1-\delta L$.
\end{remark}

\section{Multiple Remote Party Result}
In the first two sub-sections, we will implicitly assume the following: (i) We are primarily interested in whether some server was cheating and not in identifying the cheater(s) and (ii) We assume that all servers always reply back (possibly with an incorrect answer).
   \subsection{Trivial Solution}

We begin with the following direct generalization of Theorem~\ref{thm:main-single} to the multiple server case: essentially run $s$ independent copies of the protocol from Theorem~\ref{thm:main-single}.

\begin{theorem}
\label{thm:main-multiple-trivial}
For every computable error-correcting code $H: [q]^{k/s} \rightarrow [q]^n$ that is $(\rho,L)$ list-decodable, there exists an $(s,s)$-party verification protocol with resource bound $(\log{n} + s\log{q},s(\log{n} + \log{q}))$ and verification guarantee $(1-\rho,f)$, where  for any $x\in [q]^k$, $f(x)=C(x) -s -\log (s^2qL^sn^4) -2\log\log(qn) -c_0$, for some fixed positive integer $c_0$.

\end{theorem}

\begin{proof}We begin by specifying the protocol. In the pre-processing step, the client $U$ does the following
on input $x\in [q]^k$:

\begin{enumerate}
\item Generate a random $\beta\in [n]$.
\item Store $(\beta,\gamma_1=H(x_1)_{\beta},\dots,\gamma_s=H(x_s)_{\beta})$ and send $x_i$ to the server $i$ for
every $i\in [s]$.
\end{enumerate}

Server $i$ on receiving $x$, saves a string $y_i\in [q]^*$. The server is allowed to use any computable function
to obtain $y_i$ from $x_i$.

During the verification phase, $U$ does the following:

\begin{enumerate}
\item It sends $\beta$ to all $s$ servers.
\item It receives $a_i\in [q]$ from  server $i$ for every $i\in [s]$. ($a_i$ is supposed to be $H(x_i)_{\beta}$.)
\item It outputs $1$ (i.e. none of the servers ``cheated'') if $a_i=\gamma_i$ for every $i\in [s]$, else it outputs a $0$.
\end{enumerate}

Similar to the one-party result, we assume that server $i$, on receiving the challenge, uses a computable function
$\A_{x,i}:[n]\times [q]^*\rightarrow [q]$ to compute $a_i=\A_x(\beta,y_i)$ and sends $a_i$ back to $U$.

The claim on the resource usage follows immediately from the protocol specification. Next we prove its verification
guarantee. Let $T\subseteq [s]$ be the set of colluding servers. We will prove that $y_T$ is large by contradiction:
if not, then using the list decodability of $H$, we will present a description of $x_T$ of size $<C(x_T)$.
Consider the following algorithm that uses $y_T$ and an advice string $v\in\left(\{0,1\}^{|L|}\right)^{|T|}$, which
is the concatenation of shorter strings $v_{i} \in \left(\{0,1\}^{|L|}\right)$ for each $i\in T$:

\begin{enumerate}
\item Compute a description of $H$ from $n,k,\rho n,q$ and $s$.
\item For every $j\in T$, compute $z_j=(\A_{x,j}(\beta,y_j))_{\beta\in [n]}$.
\item Do the following for every $j\in T$: by cycling through all $x_j\in [q]^{k/s}$, retain the
   set $\mathcal{L}_j\subseteq [q]^{k/s}$ such that for every $u\in \mathcal{L}_j$, $\Delta(H(u),z_j) \le \rho n$.
\item For each $j \in T$, let $w_j$ be the $v_j$th string from $\mathcal{L}_j$.
\item Output the concatenation of $\{w_j\}_{j\in T}$.
\end{enumerate}

Note that since $H$ is $(\rho,L)$-list decodable, there exists an advice string $v$ such that the algorithm above
outputs $x_T$. Further, since $H$ is computable, there is an algorithm $\mathcal E$ that can compute a description
of $H$ from $n,k\rho n, q$ and $s$. (Note that using this description, we can generate any codeword $H(u)$ in step 3.)
Thus, we have description of $x_T$ of size
$|y_T|+|v|+\sum_{j\in T}|\A_{x,j}|+|\mathcal E|+ (s +\log (s^2qL^sn^4) +2\log\log(qn) +3)$
(where the term in parentheses is for encoding the different parameters and $T$), which means that
if $|y_T| < C(x_T)-|v|-\sum_{j\in T}|\A_{x,j}|-|\mathcal E| - (s +\log (s^2qL^sn^4) +2\log\log(qn) +3) = f(x)$,
then we have a description of $x_T$ of size $<C(x_T)$, which is a contradiction.
\end{proof}



\subsection{Multiple Parties, One Hash}

One somewhat unsatisfactory aspect of Theorem~\ref{thm:main-multiple-trivial} is that the storage needed by $U$ goes up a factor of $s$ from that in Theorem~\ref{thm:main-single}. Next we show that if the code $H$ is linear (and list decodable) then we can get a similar guarantee as that of Theorem~\ref{thm:main-multiple-trivial} except that the storage usage of $U$ remains the same as that in Theorem~\ref{thm:main-single}.

%

\begin{theorem}
\label{thm:main-multiple-linear}
For every computable linear error-correcting code $H: \F_q^k \rightarrow \F_q^n$ that is $(\rho,L)$ list-decodable, there exists an $(s,s)$-party verification protocol with resource bound $(\log{n} + \log{q},s(\log{n} + \log{q}))$ and verification guarantee $(1-\rho,f)$, where  for any $x\in \F_q^k$, $f(x)=C(x) -s -\log (s^2qLn^4) -2\log\log(qn) -c_0$, for some fixed positive integer $c_0$.
\end{theorem}

\begin{proof}We begin by specifying the protocol. In the pre-processing step, the client $U$ does the following on input $x\in [q]^k$:

\begin{enumerate}
\item Generate a random $\beta\in [n]$.
\item Store $(\beta,\gamma=H(x)_{\beta})$ and send $x_i$ to the server $i$ for every $i\in [s]$.
\end{enumerate}

Server $i$ on receiving $x_i$, saves a string $y_i\in [q]^*$. The server is allowed to use any computable function to obtain $y_i$ from $x_i$.
For notational convenience, we will use $\hat{x}_i$ to denote the string $x_i$ extended to a string in $\F_q^k$ by adding zeros in positions that correspond to servers other than $i$.

During the verification phase, $U$ does the following:

\begin{enumerate}
\item It sends $\beta$ to all $s$ servers.
\item It receives $a_i\in [q]$ from  server $i$ for every $i\in [s]$. ($a_i$ is supposed to be $H(\hat{x}_i)_{\beta}$.)
\item It outputs $1$ (i.e. none of the servers ``cheated") if $\gamma=\sum_{i=1}^s a_i$ else it outputs a $0$.
\end{enumerate}

We assume that server $i$ on receiving the challenge, uses a computable function $\A_{x,i}:[n]\times [q]^*\rightarrow [q]$ to compute $a_i=\A_x(\beta,y_i)$ and sends $a_i$ back to $U$.

The claim on the resource usage follows immediately from the protocol specification. Next we prove its verification guarantee. Let $T\subseteq [s]$ be the set of colluding servers. We will prove that $y_T$ is large by contradiction: if not, then using the list decodability of $H$, we will present a description of $x_T$ of size $<C(x_T)$. 

For notational convenience, define $\hat{x}_T=\sum_{j\in T} \hat{x}_j$ and $\hat{x}_{\overline{T}}=\sum_{j\not\in T} \hat{x}_j$.
Consider the following algorithm that uses $y_T$ and an advice string $v\in\{0,1\}^{|L|}$:

\begin{enumerate}
\item Compute a description of $H$ from $n,k,\rho,q,s$ and $L$.
\item Compute $z=(\sum_{j\in T} \A_{x,j}(\beta,y_j))_{\beta\in [n]}$.
\item By cycling through all $x\in \F_q^k$, retain the set $\mathcal{L}\subseteq \F_q^k$ such that for every $u\in \mathcal{L}$, $\Delta(H(u),z) \le \rho n$.
\item Output the $v$th string from $\mathcal{L}$.
\end{enumerate}

To see the correctness of the algorithm above, note that for every $j\in [s]\setminus T$, $(\A_{x,j}(\beta,y_j))_{\beta\in [n]}=H(\hat{x}_j)$. Thus, if the protocol outputs $1$ with probability at least $1-\rho$, then $\delta(z,H(\hat{x}_T))\le \rho n$ ; here we used the linearity of $H$ to note that $H(\hat{x}_T)=H(x)-H(\hat{x}_{\overline{T}})$. Note that since $H$ is $(\rho,L)$-list decodable, there exists an advice string $v$ such that the algorithm above outputs $\hat{x}_T$ (from which we can easily compute $x_T$). Further, since $H$ is computable, there is an algorithm $\mathcal E$ that can compute a description of $H$ from $s,n,k,\rho n$ and $q$. 
Thus, we have a description of $x_T$ of size $|y_T|+|v|+\sum_{j\in T}|\A_{x,j}|+|\mathcal E|+ (s +\log (s^2qLn^4) +2\log\log(qn) +3)$, (where the term in parentheses is for encoding the different parameters and $T$), which means that if $|y_T| < C(x_T)-|v|-|\A_x|-|\mathcal E|-(s +\log (s^2qLn^4) +2\log\log(qn) +3) = f(x)$, then we have a description of $x_T$ of size $<C(x_T)$, which is a contradiction.
\end{proof}

\subsection{Catching the Cheaters and Handling Unresponsive Servers}
\label{sec:cheaters}


We now observe that since the protocol in Theorem~\ref{thm:main-multiple-trivial} checks each answer $a_i$ individually to see if it is the same as $\gamma_i$, it can easily handle the case when some server does not reply back at all. Additionally, if the protocol outputs a $0$ then it knows that at least one of the servers in the colluding set is cheating. (It does not necessarily identify the exact set $T$.\footnote{We assume that identifying at least one server in the colluding set is motivation enough for servers not to collude.})

However, the protocol in Theorem~\ref{thm:main-multiple-linear} cannot identify the cheater(s) and needs all the servers to always reply back. Next, using Reed-Solomon codes, at the cost of higher user storage and a stricter bound on the number of colluding servers, we show how to get rid of these shortcomings.

Recall that a Reed-Solomon code $RS:\F_q^m\rightarrow\F_q^{\l}$ can be represented as a systematic code (i.e. the first $k$ symbols in any codeword is exactly the corresponding message) and can correct $r$ errors and $e$ erasures as long as $2r+e\le \l-m$. Further, one can correct from $r$ errors and $e$ erasures in $O(\l^3)$ time. The main idea in the following result is to follow the protocol of Theorem~\ref{thm:main-multiple-trivial} but instead of storing all the $s$ hashes, $U$ only stores the parity symbols in the corresponding Reed-Solomon codeword.

\begin{theorem}
\label{thm:multiple-rs}
For every computable linear error-correcting code $H: \F_q^k \rightarrow \F_q^n$ that is $(\rho,L)$ list-decodable, assuming at most $e$ servers will not reply back to a challenge, there exists an $(r,s)$-party verification protocol with resource bound $(\log{n} +(2r+e)\cdot \log{q}),s(\log{n} + \log{q}))$ and verification guarantee $(1- \rho,f)$, where  for any $x\in \F_q^k$, $f(x)=C(x) -s -\log (s^2qLn^4) -2\log\log(qn) -c_0$, for some fixed positive integer $c_0$.

\end{theorem}

\begin{proof}
We begin by specifying the protocol. As in the proof of
Theorem~\ref{thm:main-multiple-linear}, define $\hat{x}_i$, for $i\in [s]$,
to be the string $x_i$ extended to the vector in $\F_q^k$, which has zeros
in the positions that do not belong to server $i$. Further, for any subset
$T\subseteq [s]$, define $\hat{x}_T=\sum_{i\in T}\hat{x}_i$. Finally let
$RS:\F_q^s\rightarrow\F_q{\l}$ be a systematic Reed-Solomon code where
$\l=2r+e+s$.

In the pre-processing step, the client $U$ does the following on input
$x\in [q]^k$:

\begin{enumerate}
\item Generate a random $\beta\in [n]$.
\item Compute the vector 
$v=(H(\hat{x}_1)_{\beta},\dots,H(\hat{x}_s)_{\beta})\in\F_q^s$.
\item Store $(\beta,\gamma_1=RS(v)_{s+1},\dots,\gamma_{2r+e}=RS(v_{\l}))$
and send $x_i$ to the server $i$ for every $i\in [s]$.
\end{enumerate}

Server $i$ on receiving $x_i$, saves a string $y_i\in [q]^*$. The server is
allowed to use any computable function to obtain $y_i$ from $x_i$.


During the verification phase, $U$ does the following:

\begin{enumerate}
\item It sends $\beta$ to all $s$ servers.
\item For each server $i\in [s]$, it either receives no response or
receives $a_i\in \F_q$. ($a_i$ is supposed to be $H(\hat{x}_i)_{\beta}$.)
\item It computes the received word $z\in\F_q^{\l}$, where for $i\in [s]$,
$z_i=?$ (i.e. an erasure) if the $i$th server does not respond else
$z_i=a_i$ and for $s<i\le \l$, $z_i=\gamma_i$.
\item Run the decoding algorithm for $RS$ to compute the set 
$T'\subseteq [s]$ to be the error locations. (Note that by Step 2,
$U$ already knows the set $E$ of erasures.)
\end{enumerate}

We assume that server $i$ on receiving the challenge, uses a computable
function $\A_{x,i}:[n]\times [q]^*\rightarrow [q]$ to compute
$a_i=\A_x(\beta,y_i)$ and sends $a_i$ back to $U$ (unless it decides not
to respond).

The claim on the resource usage follows immediately from the protocol
specification. We now prove the verification guarantee. Let $T$ be the
set of colluding servers. We will prove that with probability at least
$1-\rho$, $U$ using the protocol above computes
$\emptyset\neq T'\subseteq T$ (and $|y_T|$ is large enough). Fix a
$\beta\in [n]$. If for this $\beta$, $U$ obtains $T'=\emptyset$, then
this implies that for every $i\in [s]$ such that server $i$ responds, we
have $a_i=H(\hat{x}_i)_{\beta}$. This is because of our choice of $RS$,
the decoding in Step 4 will return $v$ (which in turn allows us to compute
exactly the set $T'\subseteq T$ such that for every 
$j\in T'$, $a_j\neq H(\hat{x})_{\beta}$).\footnote{We will assume that 
$T\cap E=\emptyset$. If not, just replace $T$ by $T\setminus E$.} Thus, 
if the protocol outputs a $T'\neq\emptyset$ with probability at least 
$1-\rho$ over the random choices of $\beta$, then using the same argument as
in the proof of Theorem~\ref{thm:main-multiple-linear}, we note that 
$\Delta(H(\hat{x}_T),(\sum_{j\in T}\A_{x,j}(\beta,y_j))_{\beta\in [n]})\le \rho n$.
Again, using the same argument as in the proof of Theorem~\ref{thm:main-multiple-linear} this
implies that $|y_T|\ge C(x_T) -s -\log (s^2qLn^4) -2\log\log(qn) -c_0$, for some fixed positive
integer $c_0$.
\end{proof}

\section{Corollaries}

We now present specific instantiations of list decodable codes $H$ to obtain corollaries of our main results.

   \subsection{Optimal Storage Enforcement}
   
We begin with the following observation: If the reply from a server comes from a domain of size $q$, then one cannot hope to have a verification protocol with verification guarantee $(\delta,f)$ for any $\delta\le 1/q$ for any non-trivial $f$. This is because the server can always return a random value and no matter what function $U$ uses to compute its final output, the server will always get a favorable response with probability $1/q$.\footnote{For the single server case, the server can take $y$ to be the empty string. For the multiple server case where $T\subseteq [s]$ is the colluding set of servers, the colluding servers can, for example, ensure that they answer correctly for all but one server $i\in T$ and not store anything for server $i$.}

Next, we show that we can get $\delta$ to be arbitrarily close to $1/q$ while still obtaining $f(x)$ to be very close to $C(x)$. 
We start off with the following result due to Zyablov and Pinsker:
\begin{theorem}[\cite{ZP}]
Let $q\ge 2$ and let $0<\rho<1-1/q$. There exists a $(\rho,L)$-list decodable code with rate $1-H_q(\rho)-1/L$.
\end{theorem}

It is known that for $\eps<1/q$, $H_q(1-1/q-\eps)\le 1- C_q\eps^2$, where $C_q=q/(4\ln{q})$~\cite[Chap. 2]{a-thesis}. This implies that there exists a code $H:\F_q^k\rightarrow \F_q^n$, with $n\le \frac{k}{(C_q-1)\eps^2}\le 8k\ln{q}/(q\eps^2)$, which is $(1-1/q-\eps,1/\eps^2)$-list decodable. 
Note that the above implies that one can deterministically compute a uniquely-determined such code by iterating over all possible codes with dimension $k$ and block length $n$ and outputting the lexicographically least such one that is $(1-1/q-\eps,L)$-list decodable with the smallest discovered value of $L$.
Applying this to Theorem~\ref{thm:main-multiple-trivial} implies the following result:

\begin{corollary}
For every $\eps<1/q$ and integer $s\ge 1$,
there exists an $(s,s)$-party verification protocol with resource bound $(\log{k}+(s-1)\log{q}-2\log\eps+\frac{3}{2}\log\log{q}+3),s(\log{k}-2\log\eps+\frac{3}{2}\log\log{q}+3))$ and verification guarantee $(1/q+\eps,f)$, where for any $x\in[q]^k$, $f(x)=C(x) - s -\log{s^2k^4} +\log\eps^{2s+8} -\log\log{q^6k^2} +\log\log\eps^4 -\log\log\log{q^3} -c_0$ for some fixed positive integer $c_0$.
\end{corollary}

Some of our results need $H$ to be linear.
To this end, we will need the following result due to Guruswami et al.\footnote{The corresponding result for general codes has been known for more than thirty years.}

\begin{theorem}[\cite{GHK11}]
Let $q\ge 2$ be a prime power and let $0<\rho<1-1/q$. Then a random linear code of rate $1-H_q(\rho)-\eps$ is $(\rho,C_{\rho,q}/\eps)$-list decodable for some term $C_{\rho,q}$ that just depends on $\rho$ and $q$.
\end{theorem}

As a Corollary the above implies (along with the arguments used earlier in this section) that there exists a linear code $H:\F_q^k\rightarrow\F_q^n$ with $n\le 8k\ln{q}/(q\eps^2)$ that is $(1-1/q-\eps,C'_{\eps,q}/\eps^2)$-list decodable (where $C'_{\eps,q}\stackrel{def}=C_{1-1/q-\eps,q}$). Applying this to Theorem~\ref{thm:multiple-rs} gives us the following:

\begin{corollary}
For every $\eps<1/q$, integer $s\ge 1$, and $r,e\le s$, assuming at most $e$ servers do not reply back to a challenge,
there exists an $(r,s)$-party verification protocol with resource bound $(\log{kq^{2r+e-1}}-\log{\eps^2}+\log\log{q^{3/2}+3},s(\log{k}-\log{\eps^2}+\log\log{q^{3/2}}+3))$ and verification guarantee $(1/q+\eps,f)$, where for any $x\in[q]^k$, $f(x)=C(x) -s -\log{s^2C'_{\eps,q}k^4}+\log{q^3\eps^{10}}-\log\log{q^6k^2}+\log\log{\eps^2}-\log\log\log{q^2}-c_0$ for some fixed positive integer $c_0$.

\end{corollary}

\subsection{Practical Storage Enforcement}

All of our results so far have used computable codes $H$, which are not that useful in practice. What we really want in practice is to use codes $H$ that lead to an efficient implementation of the protocol. At the very least, all the honest parties in the verification protocol should not have to use more than polynomial time to perform the required computation. An even more desirable property would be for honest parties to be able to do their computation in a one pass, logspace, data stream fashion. In this section, we'll see one example of each. Further, it turns out that the resulting hash functions are classical ones that are also used in practice.

\subsubsection{Johnson Bound}

Before we instantiate $H$ with specific codes, we first state a general combinatorial result for list decoding codes with large distance, which will be useful in our subsequent corollaries. The result below allows for a sort of non-standard definition of codes, where a codeword is a vector in $\prod_{i=1}^n [q_i]$, where the $q_i$'s can be distinct.\footnote{We're overloading the product operator $\prod$ here to mean the iterated Cartesian product.} (So far we have looked only at the case where $q_i=q$ for $i\in [n]$.) The notion of Hamming distance still remains the same, i.e. the number of positions that two vectors differ in. (The syntactic definitions of the distance of a code and the list decodability of a code remain the same.) We will need the following result:

\begin{theorem}[\cite{G-thesis}] 
\label{thm:johnson}
Let $C$ be a code with block length $n$ and distance $d$ where the $i$th symbol in a codeword comes from $[q_i]$. Then the code is $\left(1-\sqrt{1-\frac{d}{n}},2\sum_{i=1}^n q_i\right)$-list decodable.
\end{theorem}

   \subsubsection{Hashing Modulo a Random Prime}

We will begin with a code that corresponds to the classical Karp-Rabin hash~\cite{KarpRabin}. Let $H$ be the so called Chinese Remainder Theorem (or CRT) codes. In particular, we will consider the following special case of such codes. Let $p_1\le p_2\le \cdots\le p_n$ be the first $n$ primes. Consider the CRT code $H:\prod_{i=1}^k [p_i]\rightarrow \prod_{i=1}^n [p_i]$, where the message $x\in \{0,1,\dots,(\prod_{i=1}^k p_i)-1\}$, is mapped to the vector $( x\mod p_1,x\mod p_2, \dots, x\mod p_n)\in \prod_{i=1}^n [p_i]$. It is known that such codes have distance $n-k+1$ (cf.~\cite{G-thesis}). By a simple upper bound on the prime counting function (cf.~\cite{bach-shallit}), we can take $p_n \le 2 n\log{n}$. Moreover, $\sum_{i=1}^n p_i < {np_n}/2$ (cf.~\cite{Rosser-Schoenfeld}). Thus, if we pick a CRT code with $n= k/\eps^2$, then by Theorem~\ref{thm:johnson}, $H$ is $(1-\eps, k^2(\log k - \log{\eps^2})/\eps^4)$-list decodable.

Further, note that given any $x\in \{0,1,\dots,(\prod_{i=1}^k p_i)-1\}$ and a random $\beta\in [n]$, $H(x)_{\beta}$ corresponds to the Karp-Rabin fingerprint (modding the input integer with a random prime). Further, $H(x)_{\beta}$ can be computed in polynomial time.

Thus, letting $H$ be the CRT code in Theorem~\ref{thm:main-multiple-trivial}, we get the following:
\begin{corollary}
For every $\eps>0$, there exists an $(s,s)$-party verification protocol with resource bound $$\left(
(s+1)\log{k} + s - \log{\eps^4} + s\log\log({k/\eps^2}), s\log{k}-s\log{\eps^2} + s + s\log\log{(k/\eps^2)}\right)$$
with verification guarantee $(\eps,f)$, where for every 
$x\in\{0,1,\dots,\prod_{i=1}^k p_i-1\}$,
$$f(x)=C(x) -c_0\left(s(\log(k/\eps)-\log\log(k/\eps)-1)\right) -c_1\log\log\log(k/\eps) -c_2$$
for some fixed positive integers $c_0,c_1$ and $c_2$. Further, all honest parties can do their computation in $\mathrm{poly}(n)$ time.
\end{corollary}

\begin{remark}
Theorem~\ref{thm:multiple-rs} can be extended to handle the case where the symbols in codewords of $H$ are of different sizes. However, for the sake of clarity we refrain from applying CRT to the generalization of Theorem~\ref{thm:multiple-rs}. Further, the results in the next subsection allow for a more efficient implementation of the computation required from the honest parties.
\end{remark}

   \subsubsection{Reed-Solomon Codes}

Finally, we take $H:\F_q^k\rightarrow \F_q^n$ to be the  Reed-Solomon code, with $n=q$. Recall that for such a code, given message $x=(x_0,\dots,x_{k-1})\in\F_q^k$, the codeword is given by $H(x)=(P_x(\beta))_{\beta\in\F_q}$, where $P_x(Y)=\sum_{i=0}^{k-1} x_iY^i$. It is well-known that such a code $H$ has distance $n-k+1$. Thus, if we pick $n=k/\eps^2$, then by Theorem~\ref{thm:johnson}, $H$ is $(1-\eps,2k^2/\eps^4)$-list decodable.

Further, note that given any $x\in \F_q^k$ and a random $\beta\in [n]$, $H(x)_{\beta}$ corresponds to the widely used ``polynomial" hash. Further, $H(x)_{\beta}$ can be computed in one pass over $x$ with storage of only a constant number of $\F_q$ elements. (Further, after reading each entry in $x$, the algorithm just needs to perform one addition and one multiplication over $\F_q$.)

Thus, applying $H$ as the Reed-Solomon code to Theorems~\ref{thm:main-multiple-trivial} and~\ref{thm:multiple-rs} implies the following:
\begin{corollary}
For every $\eps>0$,
\begin{itemize}
\item[(i)] There exists an $(s,s)$-party verification protocol with resource bound $( (s+1)(2\log{k}+4\log(1/\eps)+1), 2s(2\log{k}+4\log(1/\eps)+1))$ and verification guarantee $(\eps,f)$, where for any $x\in \F_q^k$, $f(x)=C(x)-O(s(\log{k}+\log(1/\eps)))$.
\item[(ii)] Assuming at most $e$ servers do not respond to challenges, there exists an $(r,s)$-party verification protocol with resource bound $( (2r+e+1)(2\log{k}+4\log(1/\eps)+1), 2s(2\log{k}+4\log(1/\eps)+1))$ and verification guarantee $(\eps,f)$, where for any $x\in \F_q^k$, $f(x)=C(x)-O(s+\log{k}+\log(1/\eps))$.
\end{itemize}
\noindent
Further, in both the protocols, honest parties can implement their required computation with a one pass, $O(\log{k}+\log(1/\eps))$ space (in bits) and $\tilde{O}(\log{k}+\log(1/\eps))$ update time data stream algorithm.
\end{corollary}



\section*{Acknowledgments} We thank Dick Lipton for pointing out the application of our protocol to the OS updating problem and for kindly allowing us to use his observation. We also thank Ram Sridhar for helpful discussions.

\bibliographystyle{abbrv}
\bibliography{verify}

\begin{thebibliography}{10}

\bibitem{AteniesePDP}
G.~Ateniese, R.~Burns, R.~Curtmola, J.~Herring, L.~Kissner, Z.~Peterson, and
  D.~Song.
\newblock Provable data possession at untrusted stores.
\newblock In {\em Proceedings of the 14th ACM conference on Computer and
  communications security}, CCS '07, pages 598--609, New York, NY, USA, 2007.
  ACM.

\bibitem{AtenieseSEPDP}
G.~Ateniese, R.~Di~Pietro, L.~V. Mancini, and G.~Tsudik.
\newblock Scalable and efficient provable data possession.
\newblock In {\em Proceedings of the 4th international conference on Security
  and privacy in communication netowrks}, SecureComm '08, pages 9:1--9:10, New
  York, NY, USA, 2008. ACM.

\bibitem{bach-shallit}
E.~Bach and J.~Shallit.
\newblock {\em {Algorithmic number theory, volume 1: efficient algorithms}}.
\newblock MIT Press, Cambridge, Massachusetts, 1996.
\newblock URL: {\tt http://www.math.uwaterloo.ca/~shallit/ant.html}.

\bibitem{BowersHAIL}
K.~D. Bowers, A.~Juels, and A.~Oprea.
\newblock Hail: a high-availability and integrity layer for cloud storage.
\newblock In {\em Proceedings of the 16th ACM conference on Computer and
  communications security}, CCS '09, pages 187--198, New York, NY, USA, 2009.
  ACM.

\bibitem{BowersPOR}
K.~D. Bowers, A.~Juels, and A.~Oprea.
\newblock Proofs of retrievability: theory and implementation.
\newblock In {\em Proceedings of the 2009 ACM workshop on Cloud computing
  security}, CCSW '09, pages 43--54, New York, NY, USA, 2009. ACM.

\bibitem{CTY10}
G.~Cormode, J.~Thaler, and K.~Yi.
\newblock Verifying computations with streaming interactive proofs.
\newblock {\em Electronic Colloquium on Computational Complexity (ECCC)},
  17:159, 2010.

\bibitem{CurtmolaMRPDP}
R.~Curtmola, O.~Khan, R.~Burns, and G.~Ateniese.
\newblock Mr-pdp: Multiple-replica provable data possession.
\newblock {\em Distributed Computing Systems, International Conference on},
  0:411--420, 2008.

\bibitem{DVW09}
Y.~Dodis, S.~P. Vadhan, and D.~Wichs.
\newblock Proofs of retrievability via hardness amplification.
\newblock In {\em Proceedings of the 6th Theory of Cryptography Conference
  (TCC)}, pages 109--127, 2009.

\bibitem{ErwayDPDP}
C.~Erway, A.~K\"{u}p\c{c}\"{u}, C.~Papamanthou, and R.~Tamassia.
\newblock Dynamic provable data possession.
\newblock In {\em Proceedings of the 16th ACM conference on Computer and
  communications security}, CCS '09, pages 213--222, New York, NY, USA, 2009.
  ACM.

\bibitem{GazzoniDDP}
E.~L. Gazzoni, D.~Luiz, G.~Filho, P.~S\'{e}rgio, and L.~M. Barreto.
\newblock Demonstrating data possession and uncheatable data transfer.
\newblock Cryptology ePrint Archive, Report 2006/150, 2006.

\bibitem{GJM02}
P.~Golle, S.~Jarecki, and I.~Mironov.
\newblock Cryptographic primitives enforcing communication and storage
  complexity.
\newblock In {\em Proceedings of the 6th International Conference on Financial
  Cryptography}, pages 120--135, 2002.

\bibitem{GolleSEC}
P.~Golle, S.~Jarecki, and I.~Mironov.
\newblock Cryptographic primitives enforcing communication and storage
  complexity.
\newblock In {\em Proceedings of the 6th international conference on Financial
  cryptography}, FC'02, pages 120--135, Berlin, Heidelberg, 2003.
  Springer-Verlag.

\bibitem{GoodsonEBT}
G.~R. Goodson, J.~J. Wylie, G.~R. Ganger, and M.~K. Reiter.
\newblock Efficient byzantine-tolerant erasure-coded storage.
\newblock In {\em Proceedings of the 2004 International Conference on
  Dependable Systems and Networks}, pages 135--, Washington, DC, USA, 2004.
  IEEE Computer Society.

\bibitem{G-thesis}
V.~Guruswami.
\newblock {\em List decoding of error-correcting codes}.
\newblock Number 3282 in Lecture Notes in Computer Science. Springer, 2004.

\bibitem{GHK11}
V.~Guruswami, J.~H{\aa}stad, and S.~Kopparty.
\newblock On the list-decodability of random linear codes.
\newblock {\em IEEE Transactions on Information Theory}, 57(2):718--725, 2011.

\bibitem{JuelsPOR}
A.~Juels and B.~S.~K. Jr.
\newblock Pors: proofs of retrievability for large files.
\newblock In P.~Ning, S.~D.~C. di~Vimercati, and P.~F. Syverson, editors, {\em
  ACM Conference on Computer and Communications Security}, pages 584--597. ACM,
  2007.

\bibitem{KarpRabin}
R.~M. Karp and M.~O. Rabin.
\newblock Efficient randomized pattern-matching algorithms.
\newblock {\em IBM J. Res. Dev.}, 31:249--260, March 1987.

\bibitem{ReedRS}
I.~S. Reed and G.~Solomon.
\newblock Polynomial codes over certain finite fields.
\newblock {\em Journal of the Society for Industrial and Applied Mathematics},
  8(2):300--304, 1960.

\bibitem{Rosser-Schoenfeld}
J.~B. Rosser and L.~Schoenfeld.
\newblock Sharper bounds for the chebyshev functions $\theta(x)$ and $\psi(x)$.
\newblock {\em Mathematics of Computation}, 29(129):pp. 243--269, 1975.

\bibitem{a-thesis}
A.~Rudra.
\newblock {\em List Decoding and Property Testing of Error Correcting Codes}.
\newblock PhD thesis, University of Washington, 2007.

\bibitem{SchwarzSFC}
T.~Schwarz and E.~L. Miller.
\newblock Store, forget, and check: Using algebraic signatures to check
  remotely administered storage.
\newblock In {\em Proceedings of the IEEE International Conference on
  Distributed Computing Systems (ICDCS '06)}, july 2006.

\bibitem{ShachamCPOR}
H.~Shacham and B.~Waters.
\newblock Compact proofs of retrievability.
\newblock In J.~Pieprzyk, editor, {\em Advances in Cryptology - ASIACRYPT
  2008}, volume 5350 of {\em Lecture Notes in Computer Science}, pages 90--107.
  Springer Berlin / Heidelberg, 2008.

\bibitem{ShahPPAUDIT}
M.~A. Shah, R.~Swaminathan, and M.~Baker.
\newblock Privacy-preserving audit and extraction of digital contents.
\newblock Technical Report HPL-2008-32R1, HP Labs, 2008.

\bibitem{WangEDS}
C.~Wang, Q.~Wang, K.~Ren, and W.~Lou.
\newblock Ensuring data storage security in cloud computing.
\newblock Cryptology ePrint Archive, Report 2009/081, 2009.

\bibitem{WangEPVDD}
Q.~Wang, C.~Wang, J.~Li, K.~Ren, and W.~Lou.
\newblock Enabling public verifiability and data dynamics for storage security
  in cloud computing.
\newblock In M.~Backes and P.~Ning, editors, {\em Computer Security ESORICS
  2009}, volume 5789 of {\em Lecture Notes in Computer Science}, pages
  355--370. Springer Berlin / Heidelberg, 2009.

\bibitem{ZP}
V.~V. Zyablov and M.~S. Pinsker.
\newblock List cascade decoding.
\newblock {\em Problems of Information Transmission}, 17(4):29--34, 1981 (in
  Russian); pp. 236-240 (in English), 1982.

\end{thebibliography}

\appendix

\section{Tables}
\label{Appendix:tables}

\begin{table*}[htbp]
  \centering
  \caption{Summary of the existing schemes compared to the proposed scheme, AIS. \textit{Method} denotes the technique primarily used for data possession verification by a scheme. The row \textit{proof of storage enforcement} refers to whether a scheme forces a server to store as much as the original data. The capability of a trusted third party to verify possession of data is the \textit{public verifiability} metric. \textit{Retrievability} denotes the capability to retrieve the original data from the responses of the server. Finally, \textit{corruption detection} refers to whether the verification is for complete data or partial data. When complete data is verified, any modification of data will be detected whereas in partial verification, a small fraction of corruption can remain undetected. The \textit{Unlimited audit} row refers to whether the scheme is for bounded or unbounded usage. The term NEC refers to No Explicit Construction. Please note that there are different variations of POR and PDP in the existing literature. For the sake of brevity, we only consider the original proposals in this comparison table, whereas other variations are discussed in the related works section. }
  	\scalebox{0.8}{
    \begin{tabular}{cccccccccc}
    \addlinespace
    \toprule
    \textit{Metric/Scheme} & \textbf{SFC \cite{SchwarzSFC}} & \textbf{POR \cite{DVW09}} & \textbf{SDS \cite{WangEDS}} & \textbf{HAIL \cite{BowersHAIL}} & \textbf{DDP \cite{GazzoniDDP}} & \textbf{PDP \cite{AteniesePDP}} & \textbf{EPV \cite{WangEPVDD}} & \textbf{SEC \cite{GolleSEC}} & \textbf{AIS} \\
    \midrule
    \textit{Method} & Code  & Code  & Code  & Code  & Crypto & Crypto & Crypto & Crypto & Code \\
    \textit{Proof of Storage Enforcement} & NEC   & NEC   & NEC   & NEC   & NEC   & NEC   & NEC   & Yes   & Yes \\
    \textit{Public Verifiability} & Yes   & NEC   & Yes   & Yes   & Yes   & Yes   & Yes   & Yes   & Yes \\
    \textit{Retrievability} & NEC   & Yes   & NEC   & NEC   & NEC   & NEC   & Yes   & NEC   & Yes \\
    \textit{Corruption Detection} & Complete & Partial & Partial & Partial & Complete & Partial & Partial & Partial & Complete \\
    \textit{Unlimited Audit} & Yes   & No    & No    & Yes   & Yes   & Yes   & Yes   & Yes   & No \\
    \bottomrule
    \end{tabular}%
    }
  \label{tab:schemesummary}%
\end{table*}%

Existing approaches for data possession verification at remote storage can be broadly classified into two categories: Crypto-based and Coding based. \textit{Crypto-based approaches} rely on symmetric and assymetric cryptographic primitives for proof of data possession. Ateniese et al. \cite{AteniesePDP} defined the proof of data possession (PDP) model which uses public key homomorphic tags for verification of stored files. It can also support public verifiability with a slight modification of the original protocol by adding extra communication cost. In subsequent work, Ateniese et al. \cite{AtenieseSEPDP} proposed a symmetric crypto-based variation (SEP) which is computationally efficient compared to the original PDP but lacks public verifiability. Also, both of these protocols considered the scenario with files stored on a single server, and do not discuss erasure tolerance. However, Curtmola et al. \cite{CurtmolaMRPDP} extended PDP to a multiple-server scenario by introducing multiple identical replicas of the original data. Among other notable constructions of PDP, Gazzoni et al.\cite{GazzoniDDP} proposed a scheme (DDP) that relied on an RSA-based hash (exponentiating the whole file), and Shah et al. \cite{ShahPPAUDIT} proposed a symmetric encryption based storage audit protocol. Recent extensions on crypto-based PDP schemes by Wang et al. (EPV) \cite{WangEPVDD} and Erway et al. \cite{ErwayDPDP} mainly focus on supporting data dynamics in addition to existing capabilities. Golle et al. \cite{GolleSEC} had proposed a cryptographic primitive called storage enforcing commitment (SEC) which probabilistically guarantees that the server is using storage whose size is equal to the size of the original data to correctly answer the data possession queries. In general, the drawbacks of the aforementioned protocols are: \textit{(a)} being computation intensive due to the usage of expensive cryptographic primitives and \textit{(b)} since each verification checks a random fragment of the data, a small fraction of data corruption might go undetected and hence they do not guarantee the retrievability of the original data. \textit{Coding-based approaches}, on the other hand, have relied on special properties of linear codes such as the Reed-Solomon (RS) \cite{ReedRS} code. The key insight is that encoding the data imposes certain algebraic constraints on it which can be used to devise efficient fingerprinting scheme for data verification. Earlier schemes proposed by Schwarz et al. (SFC) \cite{SchwarzSFC} and Goodson et al. \cite{GoodsonEBT} are based on this and are primarily focused on the construction of fingerprinting functions and categorically falls under distributed protocols for file integrity checking. Later, Juels and Kaliski \cite{JuelsPOR} proposed a construction of a proof of retrivability (POR) which guarantees that if the server passes the verification of data possession, the original data is retrivable with high probability. While the scheme by Juels \cite{JuelsPOR}, supported a limited number of verifications, the theoretical POR construction by Shacham and Waters \cite{ShachamCPOR} extended it to unlimited verification and public verifiability by integrating cryptographic primitives. Subsequently, Dodis et al. \cite{DVW09} provided theoretical studies on different variants of existing POR scheme and Bowers et al. \cite{BowersPOR} considered POR protocols of practical interest \cite{JuelsPOR, ShachamCPOR} and showed how to tune parameters to achieve different performance goals. However, these POR schemes only consider the single server scenario and have no construction of a retrivability guarantee in a distributed storage scenario.
Very recently, protocols developed by Wang et al. (SDS) \cite{WangEDS} and Bowers et al. (HAIL) \cite{BowersHAIL} focus on securing distributed cloud storage in terms of availability and integrity.
Table \ref{tab:schemesummary} summarizes and compares existing schemes with AIS. Asymptotic complexity of different operations are compared in Table \ref{tab:addlabel2}.\\

\begin{table*}[htbp]
  \centering
  \caption{Asymptotic performance comparison of existing schemes with the proposed scheme. We assume that data contains $n$ symbols each of size $logq$ bits, that it is then divided into $s$ equal-sized blocks, with the blocks being distributed among the servers. We compare the token generation and verification for all $s$ blocks. $\xi$ is the fraction of symbols checked during each verification for the schemes which check partial data corruption. For cryptographic schemes, we assume $E_m$ is the cost of performing modular exponentiation modulo $m$. Token generation/verification complexity is based on the number of bit operations. Storage and communication complexity is based on the number of bits. AIS-S refers to our proposed simple scheme where we generate one token for each server and AIS-E is a variation where a single token can verify multiple servers. Additional server storage refers to the amount of data that the server stores in addition to the original data, if any.}
  \scalebox{0.6}{
    \begin{tabular}{ccccccccccc}
    \addlinespace
    \toprule
    \textit{Operation/Scheme} & \textbf{SFC \cite{SchwarzSFC}} & \textbf{POR \cite{JuelsPOR}} & \textbf{SDS \cite{WangEDS}} & \textbf{HAIL \cite{BowersHAIL}} & \textbf{DDP \cite{GazzoniDDP}} & \textbf{PDP \cite{AteniesePDP}} & \textbf{EPV \cite{WangEPVDD}} & \textbf{SEC \cite{GolleSEC}} & \textbf{AIS-S} & \textbf{AIS-E}  \\

    \midrule
    Token Generation & $O (nlogq)$ & $O ((n^2)logq)$ & $O (nlogq)$ & $O (nlognlogq)$ & $O(nE_m)$ & $O(nE_m)$ & $O(nE_m)$ & $O(nE_m)$ & $O (nlogq)$ & $O ((n/s) logq)$ \\
    Proof Generation & $O (nlogq)$ & $O((n/\xi) logq)$ & $O((n/\xi) logq)$ & $O((n/\xi) logq)$ & $O(nE_m)$ & $O((n/\xi) E_m)$ & $O((n/\xi) E_m)$ & $O((n/\xi) E_m)$ & $O (nlogq)$ & $O (nlogq)$ \\
    Proof Verification & $O(nlognlogq)$ & $O(1)$ & $O(nlognlogq)$ & $O (nlognlogq)$ & $O(nE_m)$ & $O((n/\xi) E_m)$ & $O((n/\xi) E_m)$ & $O((n/\xi) E_m)$ & $O(1)$ & $O(1)$ \\
    Client Storage & $O(1)$ & $O(slogq)$ & $O(slogq)$ & $O(1)$ & $O(1)$ & $O(1)$ & $O(slogm)$ & $O(slogm)$ & $O(slogq)$ & $O(1)$ \\
    Add. Server Storage & $0$ & $0$ & $0$ & $O(nlogq)$ & $0$ & $O(slogm)$ & $O(1)$ & $O(slogm)$ & $0$ & $0$ \\
    Communication Complexity & $O(slogq)$ & $O((n/\xi) logq)$ & $O((n/\xi) logq)$ & $O((n/\xi) logq)$ & $O(slogm)$ & $O(slogm)$ & $O(slogm)$ & $O(slogm)$ & $O(slogq)$ & $O(slogq)$ \\
    \bottomrule
    \end{tabular}%
    }
  \label{tab:addlabel2}%
\end{table*}%

\end{document}